\documentclass[conference,a4paper]{IEEEtran}
\IEEEoverridecommandlockouts

\usepackage{cite}
\usepackage{algorithmic}
\usepackage{graphicx}
\usepackage{textcomp}
\usepackage{dsfont}
\usepackage{lipsum}
\usepackage{amssymb, amsmath, amsthm, amsfonts}
\usepackage{tikz}
\usepackage{xcolor}
\usetikzlibrary{trees}
\newtheorem{thm}{Theorem}[section]
\newtheorem{lem}[thm]{Lemma}

\ifCLASSOPTIONcompsoc
\usepackage[caption=false, font=normalsize, labelfont=sf, textfont=sf]{subfig}
\else
\usepackage[caption=false, font=footnotesize]{subfig}
\fi

\def\BibTeX{{\rm B\kern-.05em{\sc i\kern-.025em b}\kern-.08em
    T\kern-.1667em\lower.7ex\hbox{E}\kern-.125emX}}

\def\v{\mathbf{v}}

\def\E{\mathbb{E}}
\def\P{\mathbb{P}}

\begin{document}

\title{The Effect of Introducing Redundancy in a Probabilistic Forwarding Protocol}

\author{\IEEEauthorblockN{Vinay Kumar B.R.$^\dag$} \and \IEEEauthorblockN{Roshan Antony$^\ddag$} \and \IEEEauthorblockN{Navin Kashyap$^\dag$}}

\maketitle

\renewcommand{\thefootnote}{}
\footnotetext{
\noindent $^\dag$Vinay Kumar B.R.\  and N.\ Kashyap are with the Department of Electrical Communication Engineering, Indian Institute of Science, Bangalore. Email: \{vinaykb, nkashyap\}@iisc.ac.in

 $^\ddag$Roshan Antony is currently with Qualcomm India, Bangalore. Email: roshanantony@outlook.com}

\begin{abstract}
This paper is concerned with the problem of broadcasting information from a source node to every node in an ad-hoc network. Flooding, as a broadcast mechanism, involves each node forwarding any packet it receives to all its neighbours. This results in excessive transmissions and thus a high energy expenditure overall. Probabilistic forwarding or gossiping involves each node forwarding a received packet to all its neighbours only with a certain probability $p$. In this paper, we study the effect of introducing redundancy, in the form of coded packets, into a probabilistic forwarding protocol. Specifically, we assume that the source node has $k$ data packets to broadcast, which are encoded into $n \ge k$ coded packets, such that any $k$ of these coded packets are sufficient to recover the original $k$ data packets. Our interest is in determining the minimum forwarding probability $p$ for a ``successful broadcast'', which we take to be the event that the expected fraction of network nodes that receive at least $k$ of the $n$ coded packets is close to 1. We examine, via simulations and analysis of a number of different network topologies (e.g., trees, grids, random geometric graphs), how this minimum forwarding probability, and correspondingly, the expected total number of packet transmissions varies with the amount of redundancy added. Our simulation results indicate that over network topologies that are highly connected, the introduction of redundancy into the probabilistic forwarding protocol is useful, as it can significantly reduce the expected total number of transmissions needed for a successful broadcast. On the other hand, for trees, our analysis shows that the expected total number of transmissions needed increases with redundancy.
\end{abstract}


\renewcommand{\thefootnote}{\arabic{footnote}}

\section{Introduction}
An ad-hoc network is a network of nodes which communicate with each other without relying on any centralized infrastructure. They are an integral part of defence operations and rescue missions. Instances include reconnaissance by soldiers, rescue during earthquakes, surveillance using drones, Geographical Information Systems etc. 
\par Many such applications necessitate certain information to be broadcast from a source node to all the other nodes in the network. Flooding is a common strategy for broadcasting content to all nodes. In this strategy, each node, upon receiving a new message packet, forwards it to all its one-hop neighbours. While this strategy is simple and easy to implement, it is wasteful in terms of overall power consumption as the total number of transmissions across all nodes in the network can be quite high. It is also known to result in the `broadcast storm' problem \cite{tseng2002broadcast}.
\par An attractive alternative that has been considered in the literature is probabilistic retransmission \cite{sasson2003probabilistic} or probabilistic forwarding, in which each node in the network, upon receiving a new packet, decides to broadcast it to its one-hop neighbours with probability $p$, and takes no action with probability $1-p$.
\par The probabilistic forwarding algorithm has been well studied with reference to ad-hoc networks. It is also referred to as the gossip protocol in some literature; for instance, in \cite{haas2006gossip}, gossiping is compared with flooding and is found to save up to 35\% message overhead. It is also used alongside routing protocols to improve network performance in terms of end-to-end latency and throughput.
\par In our scenario, a source node needs to transmit $k$ message packets to a large fraction of nodes in the network. These $k$ message packets are encoded into $n$ coded packets and are transmitted by the source. We assume that, any node receiving at least $k$ out of these $n$ coded packets can decode the original $k$ message packets. The source transmits all $n$ coded packets with probability 1, whereas other nodes in the network employ the probabilistic forwarding algorithm. 
\par Our goal is to analyze the performance of the above algorithm. In particular, we wish to find the minimum retransmission probability $p$ for which the expected fraction of nodes receiving at least $k$ out of the $n$ coded packets is close to 1, which we deem a ``successful broadcast''. In other words, we wish to determine the probability with which every node needs to retransmit a received packet, so that (with high probability) almost all the nodes in the network can decode the $k$ message packets which the source intended to communicate. This probability yields the minimum value for the expected total number of transmissions across all the network nodes for a successful broadcast. We study the variation of the expected total number of transmissions with redundancy. 
\par Our simulation results show that, over a variety of network topologies that are highly connected --- for example, grids, and random geometric graphs above the connectivity threshold --- the expected total number of transmissions by all the nodes of the network decreases initially to a minimum and then increases as the redundancy $\rho=\frac{n-k}{k}$ is increased. This means that there is some value of redundancy which is optimal, in the sense that it minimizes the number of transmissions. Consequently, a network in the grid or the random geometric topology performs best when operated at this value of redundancy and the corresponding minimum forwarding probability. On the other hand, over trees, our simulations and analysis indicate that there is no benefit to introducing redundancy in the probabilistic forwarding protocol: the expected total number of transmissions increases with redundancy. 
\par The rest of the paper is organized as follows. In Section~\ref{sec:formulation}, we provide a theoretical framework for the problem we are trying to address. Section~\ref{sec:simulation} has simulation results for the probabilistic forwarding algorithm on different types of network topologies, such as random geometric graphs, grids and trees. In Section~\ref{sec:trees}, we present a mathematical analysis of the protocol on trees, which explains the simulation results obtained for that topology. Section~\ref{sec:discuss} discusses and gives some heuristic insight into our results.  

\section{Problem Formulation}\label{sec:formulation}
Consider a graph $G=(V,E)$ where $V$ is the vertex set with $N$ vertices (nodes) and $E$ is the set of edges (communication links). A source node $s \in V$ has $k$ message packets which need to be broadcast in the network. The source $s$ encodes the $k$ messages into $n$ coded packets using a Maximum Distance Separable (MDS) code (see e.g., \cite[Ch.~11]{roth2006}). Such a code is able to correct up to $n-k$ erasures. Thus, on receiving any $k$ of these $n$ coded packets, a node can retrieve the original $k$ message packets by treating the unreceived packets as erasures. We assume that all the required encoding/decoding operations are carried out over a sufficiently large field, so that an MDS code with the necessary parameters exists. The redundancy introduced in this scheme is $\rho = \frac{n-k}{k}$.

 The source node broadcasts all $n$ coded packets to its one-hop neighbours, after which the probabilistic forwarding protocol takes over. A node receiving a particular packet for the first time, transmits it to all its one-hop neighbours with probability $p$ and takes no action with probability $1-p$. Each packet is transmitted independently of other packets and other nodes. This probabilistic retransmission continues until the time there are no further transmissions in the system. This maximum time is finite since each node in the network decides to transmit a particular packet only the first time it is received. Subsequent receptions of that packet are ignored.
\par We are interested in the following scenario. Let $R_{k,\rho}$ be the number of nodes, including the source node, that receive at least $k$ out of the $n$ packets. Given a $\delta \in (0,1)$, let $p_{k,\rho,\delta}$ be the minimum forwarding probability $p$ such that $\mathbb{E}\left[\frac{R_{k,\rho}}{N}\right] \geq 1-\delta$. The performance measure of interest, denoted by $\tau_{k,\rho,\delta}$, is the expected total number of transmissions across all nodes when the forwarding probability is set to $p_{k,\rho,\delta}$. Here, it should be clarified that each network node transmits a given coded packet at most once; a single (broadcast) transmission of a packet by a node is received by all its one-hop neighbours. Since each network node transmits each of the $n$ coded packets independently with probability $p_{k,\rho,\delta}$, the expected number of transmissions by a particular node, assuming it receives all $n$ packets, is $np_{k,\rho,\delta}$. However, a node may not receive all $n$ packets owing to the structure of the graph. Hence, by independence of packet transmissions across nodes, and the fact that the source node braodcasts all $n$ packets, the expected total number of transmissions, $\tau_{k,\rho\delta}$, is bounded above by $n+(N-1)np_{k,\rho,\delta}.$
Our aim is to determine, for a given $k$ and $\delta$, how $\tau_{k,\rho,\delta}$ varies with $\rho$, and the value of $\rho$ at which it is minimized.

\section{Simulation results}\label{sec:simulation}
Simulations were performed on random geometric graphs, grids and binary trees. 
For each of these graphs, the value of $k$ and $\delta$ was fixed initially and transmission of $n$ packets was carried out from a source node as explained in the previous section.  The value of $p$ was decreased from 1 and the average number of nodes that receive $k$ out of the $n$ packets was computed over $500$ simulation trials. The minimum forwarding probability for which this average value exceeded $1-\delta$ was recorded as $p_{k,\rho,\delta}$ along with the corresponding value of $\tau_{k,\rho,\delta}$. Both these quantities were plotted as a function of the redundancy $\rho$. The results are summarized below.
\subsection{Random geometric graph (RGG)}
\begin{figure}
	\centering
	\includegraphics[width=0.55\linewidth]{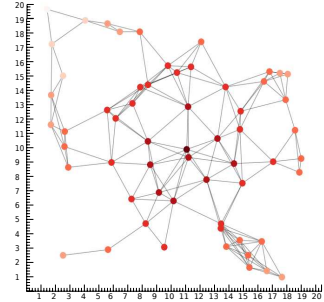}
	\caption{Random geometric graph.}
		\vspace{-1.5em}
	\label{rgg}
\end{figure}
RGGs have been widely used to model ad-hoc networks \cite{sinclair2010mobile}. In an RGG, nodes are distributed uniformly in some region, and two nodes are connected by an edge iff they are at most at a prescribed distance $r$ from each other. In our simulations, 60 nodes are deployed uniformly in a rectangular region of $20\times 20$ units, as shown in Fig. \ref{rgg}. These form the vertices of the graph. The distance $r$ is chosen so that the RGG operates well above the connectivity threshold\footnote{The connectivity threshold is the least distance $r$ for which the RGG is connected, with high probability \cite[Ch.~7, pp.~158--159]{vaze2015random}.}. A node is chosen randomly to be the source $s$, and the probabilistic forwarding algorithm is simulated. The simulation is done for $k = 100$, and $\delta=0.1$ and $0.05$, and the results are shown in Fig. \ref{fig1}.
\begin{figure} 
	\centering
	\subfloat[Minimum retransmission probability]{%
		\includegraphics[width=\linewidth]{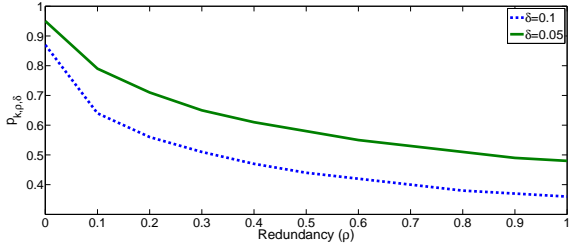}}
	\label{1a}
	\subfloat[Expected total number of transmissions]{%
		\includegraphics[width=\linewidth]{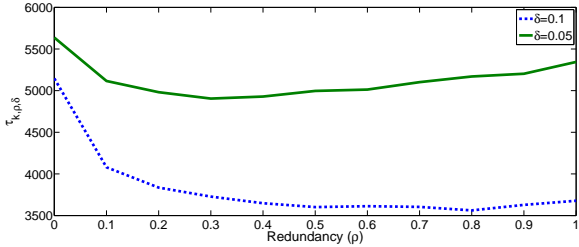}}
	\label{1b}\\
	\caption{Probabilistic forwarding on a RGG of 60 nodes in a $20\times20$ area with radius $r$ = 5.5 units.}
	\label{fig1} 
	\vspace{-1.5em}
\end{figure}
\par It is observed that, for fixed $k$ and $\delta$, the minimum retransmission probability $p_{k,\rho,\delta}$ decreases as the redundancy $\rho$ is increased. This is expected since with a larger number of coded packets $n$, we can afford to retransmit each packet with a smaller probability while still being able to deliver $k$ out of the $n$ packets to a $1-\delta$ fraction of the nodes in the network. 

A much more interesting trend is that of the expected total number of transmissions $\tau_{k,\rho,\delta}$, which initially decreases and then grows gradually as the redundancy $\rho$ is increased. There is thus an optimal value of $\rho$ that minimizes $\tau_{k,\rho,\delta}$. This happens due to an interplay between two opposing factors: an increase in $\rho$ leads to a decrease in $p_{k,\rho,\delta}$, which contributes towards a decrease in $\tau_{k,\rho,\delta}$. But this is opposed by the fact that a higher redundancy tends to increase the number of transmissions, since there are a larger number of packets to be transmitted in the network. The initial decrease in $\tau_{k,\rho,\delta}$ can be attributed to the dominant effect of the initial steep decrease in $p_{k,\rho,\delta}$. However, as the redundancy is further increased, the decrease in $p_{k,\rho,\delta}$ becomes more gradual. In this regime, as the number of coded packets $n$ increases, the gain obtained via the slight decrease in $p_{k,\rho,\delta}$ is more than offset by the fact that there are more packets to be transmitted in the network. This trend can be seen much more clearly in a grid topology.

\subsection{Two-Dimensional (2-D) Grid}
Consider a $31\times31$ square grid as shown in Fig. \ref{grids}(a). The source node is assumed to be at the center of the grid and the probabilistic forwarding algorithm is implemented. The simulation results for $k=100$ and $\delta=0.1,0.05$ are shown in Fig. \ref{fig2}.
\begin{figure} 
	\centering
	\subfloat[Minimum retransmission probability]{%
		\includegraphics[width=\linewidth]{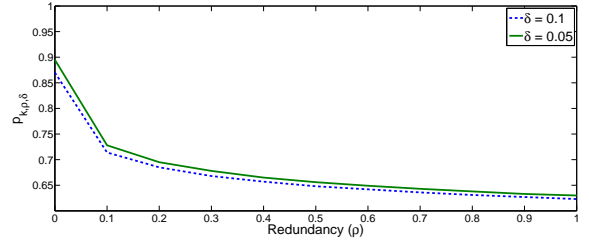}}
	\label{2a}
	\subfloat[Expected total number of transmissions]{%
		\includegraphics[width=\linewidth]{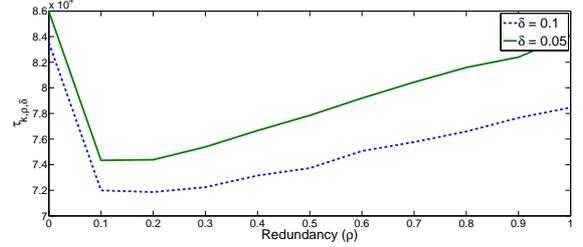}}
	\label{2b}\\
	\caption{Probabilistic forwarding on a $31\times31$ grid.}
	\label{fig2} 
	\vspace{-1.5em}
\end{figure}
The decrease in $p_{k,\rho,\delta}$ as a function of $\rho$ is similar to that seen in the RGG setting above. The variation of $\tau_{k,\rho,\delta}$ with $\rho$ is also similar to that observed in the RGG, but the trend is more pronounced here. Note that even when $\rho=0$, $p_{k,\rho,\delta}\neq1$, since we only ask for a $1-\delta$ fraction of the nodes to receive at least $k$ packets.

\subsection{Binary tree}
A rooted binary tree of height $H$ is the graph depicted in Fig.~\ref{bintree}. The probabilistic forwarding strategy with coded packets is simulated over a rooted binary tree of height 10. The root of the tree is the source node. The simulation results are shown in Fig.~\ref{fig3}, again for $k=100$ and $\delta=0.1,0.05$.
\begin{figure} 
	\centering
	\subfloat[Minimum retransmission probability]{%
		\includegraphics[width=\linewidth]{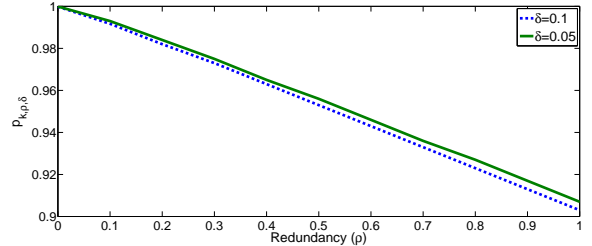}}
	\label{3a}
	\subfloat[Expected total number of transmissions]{%
		\includegraphics[width=\linewidth]{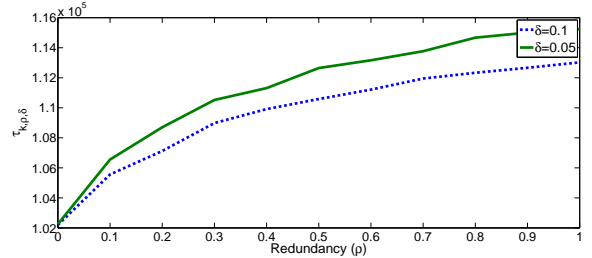}}
	\label{3b}\\
	\caption{Probabilistic forwarding on a binary tree of height $H=10$.}
	\label{fig3} 
	\vspace{-1.5em}
\end{figure}
Here, we find that the minimum probability of retransmission decreases with increase in redundancy, but the rate of decrease is much lower than that in the previous two topologies.  Consequently, the expected total number of transmissions $\tau_{k,\rho,\delta}$ always increases with redundancy. A theoretical analysis of these observations is presented in the following section.

\section{Analysis of Probabilistic Forwarding on Trees}\label{sec:trees}
In this section, we analyze the probabilistic forwarding mechanism on trees. In particular, we study rooted binary trees. However, the results in this section can be easily extended to $d$-ary trees.


\par Consider again a rooted binary tree of height $H \ge 1$ as shown in Fig. \ref{bintree}.
\begin{figure}
	\centering
	\includegraphics[width=0.7\linewidth]{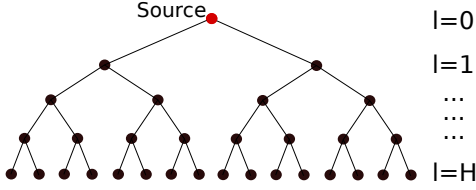}
	\caption{A rooted binary tree of height $H$.}
	\label{bintree}
	\vspace{-1.5em}
\end{figure}
 The tree consists of $H$ levels, with the root node at level $l = 0$, and for $l = 1,2,\ldots,H-1$, each node at level $l$ having two children at level $l+1$. Thus, there are $2^l$ nodes at level $l$, for $l = 0,1,2,\ldots,H$, so that the total number of nodes in the tree is $N=\sum_{l=0}^H 2^l = 2^{H+1}-1$. The root node is taken to be the source node. It initially has $k$ message packets which it encodes into $n = k(1+\rho)$ coded packets using an MDS code and transmits all of them (with probability $1$) to each of its children. Subsequent transmissions of the $n$ coded packets follow the probabilistic forwarding mechanism for some fixed value of the retransmission probability $p \in [0,1]$. Note that nodes that share a common parent receive the same packets and hence will possess the same number of packets at the end of the probabilistic forwarding mechanism. We will assume that the nodes at level $H$ (i.e., the leaf nodes) do not transmit, as there is nothing to be gained in allowing them to do so. 
 
 As explained in Section~\ref{sec:formulation}, the reception of any $k$ out of the $n$ coded packets suffices to recover the $k$ message packets. The notation $R_{k,\rho}$ was introduced there to denote the random number of nodes (including the source node) that possess at least $k$ out of the $n$ coded packets at the end of the protocol. We further let $T_{k,\rho}$ be the total number of transmissions that take place across all $N$ nodes during the course of the protocol. Our primary interest is in determining, for a given $\delta \in (0,1)$, the quantity $\tau_{k,\rho,\delta}$, which denotes the minimum expected value of $T_{k,\rho}$, minimized over all $p \in [0,1]$ such that $\E\left[\frac{R_{k,\rho}}{N}\right] \ge 1-\delta$. As we will see below, $\E[T_{k,\rho}]$ is a monotonically increasing function of $p$, so that it is also of interest to determine $p_{k,\rho,\delta}$, which is the minimum probability $p \in [0,1]$ such that $\mathbb{E}\left[\frac{R_{k,\rho}}{N}\right] \geq 1-\delta$. 

  We can express $R_{k,\rho}$ as $\sum_{l = 0}^H R_l$, where $R_l$ is the number of nodes at level $l$ that hold at least $k$ of the $n$ packets. Note that $R_0 = 1$. Similarly, $T_{k,\rho} = \sum_{l = 0}^{H-1} T_l$, where $T_l$ is the number of transmissions by nodes at level $l$. Note that $T_0 = n$, since the source node always transmits all $n$ coded packets. Also, $T_H = 0$, since leaf nodes are assumed not to transmit. We compute the expected values of $R_l$ and $T_l$ next.

Note that there is only a single path from the root to any node in the tree. Thus, for a node $\v$ at level $l$ to receive the $j$th coded packet from the root, all the intermediate nodes on the unique path from the root to $\v$ need to transmit the $j$th packet. Hence, for $l \ge 1$, 
\begin{equation}
\mathbb{P}(\text{node $\v$ at level $l$ receives the $j$th packet}) = p^{l-1}. \label{eq:vlj}
\end{equation}
Since distinct packets are transmitted independently of each other, we have
\begin{IEEEeqnarray}{rCl}\label{prob_nodes}
	\P(\text{node $\v$} & & \!\!\!\!\!\!\!\! \text{ at level $l$ receives at least $k$ out of $n$ packets}) \nonumber \\
	&=& \sum_{k'=k}^{n} \binom{n}{k'}p^{(l-1)k'}(1-p^{l-1})^{n-k'} \nonumber \\
	&=& \P(Z_{l-1} \ge k) \nonumber
\end{IEEEeqnarray}
where $Z_{l-1} \sim \text{Bin}(n,p^{l-1})$ is a binomial random variable with parameters $n$ and $p^{l-1}$. Summing the above over all nodes $\v$ at level $l$, we obtain $\E[R_l] = 2^l \,\P(Z_{l-1} \ge k)$, and hence, 
\begin{equation}\label{rCl}
\E[R_{k,\rho}] \ = \  1+\mathbb{E}\left[\sum_{l=1}^{H}R_l\right] \ = \ 1+ \sum_{l=1}^{H}2^l \, \P(Z_{l-1}\ge k).
\end{equation}

Following \eqref{eq:vlj}, we also see that the probability that a node $\v$ at level $l$ receives and retransmits the $j$th packet equals $p^l$, for $l = 1,2,\ldots,H-1$. Hence, the expected number of transmissions of the $j$th packet by nodes at level $l$ is equal to $2^l\,p^l$. Summing over $j=1,2,\ldots,n$, we obtain that $\E[T_l] = n(2p)^l = k(1+\rho)(2p)^l$, and as a consequence,
$$
\E[T_{k,\rho}] \ = \ \sum_{l=0}^{H-1} \E[T_l] \ = \  k(1+\rho)\frac{(2p)^H-1}{2p-1}.
$$
Thus, $\E[T_{k,\rho}]$ is a monotonically increasing function of $p$, from which we infer that
\begin{equation}
\tau_{k,\rho,\delta}=k(1+\rho) \frac{(2p_{k,\rho,\delta})^{H}-1}{2p_{k,\rho,\delta}-1}.
\label{taukrho}
\end{equation}
Consequently, to understand how $\tau_{k,\rho,\delta}$ behaves as a function of the redundancy $\rho$, for fixed $k$ and $\delta$, it is necessary to understand how $p_{k,\rho,\delta}$ varies with $\rho$. In the remainder of this section, we derive a good approximation for $p_{k,\rho,\delta}$ valid for all $\delta$ sufficiently close to $0$ and for all sufficiently large $k$. We will also assume that the height $H$ of the tree is large enough that $2^{H+1} \gg 1$.

From \eqref{rCl}, we see that $p_{k,\rho,\delta}$ is the least value of $p \in [0,1]$ for which 
\begin{equation}\label{exprtomin}
\frac{1}{2^{H+1}-1}+\frac{\sum_{l=0}^{H-1}2^{l+1}\mathbb{P}(Z_l\geq k)}{2^{H+1}-1} \geq 1-\delta.
\end{equation}
where $Z_l \sim \text{Bin}(k(1+\rho),p^l)$ for $l = 0,1,\ldots,H-1$. To proceed, we need the following lemma.

\begin{lem}\label{lemma}
For $\zeta \sim \text{Bin}(n,p)$ and $0 \le k \le n$, $\P(\zeta \ge k)$ is a continuous, monotonically increasing function of $p$.
\end{lem}
\begin{proof} Note first that $\P(\zeta \ge k) = \sum_{k' \ge k} \binom{n}{k} p^{k'} (1-p)^{n-k'}$, which is a polynomial in $p$, and hence $\P(\zeta \ge k)$ is continuous in $p$. Monotonicity is by a standard coupling argument: Let $U_i$, $i = 1,2,\ldots,n$, be i.i.d.\ $\text{Unif}[0,1]$ random variables. For $p \le p'$, let $X_i = \mathbb{I}_{\{U_i \le p\}}$ and $X_i' = \mathbb{I}_{\{U_i \le p'\}}$, so that the $X_i$s are i.i.d.\ Ber$(p)$ and the $X'_i$s are i.i.d\ Ber$(p')$. Then, $\zeta = \sum_{i=1}^n X_i$ is Bin$(n,p)$, while $\zeta' = \sum_{i=1}^n X'_i$ is Bin$(n,p')$. By construction, $X_i(U_i) \le X'_i(U_i)$, and hence, $\zeta \le \zeta'$ almost surely. Thus, $\P(\zeta \ge k) \le \P(\zeta' \ge k)$. \end{proof}

The lemma above shows that the LHS of \eqref{exprtomin} monotonically increases from $\frac{1}{2^{H+1}-1}$ to $1$, as $p$ goes from $0$ to $1$. Thus, if $\delta > 0$ is such that $1-\delta$ lies between $\frac{1}{2^{H+1}-1}$ and $1$, then by continuity of the LHS, $p_{k,\rho,\delta}$ is the least value of $p$ for which the inequality in \eqref{exprtomin} holds with equality. To obtain an estimate of $p_{k,p,\delta}$, we re-write \eqref{exprtomin} as
\begin{equation}\label{key_ineq}
\frac{\sum_{l=0}^{H-1}2^{l+1}\mathbb{P}(Z_l\geq k)}{2^{H+1}-1} \geq 1-\delta-\frac{1}{2^{H+1}-1},
\end{equation}
and obtain an approximation of the LHS above using the fact that $Z_l$ concentrates about its mean, $k(1+\rho)p^l$. We will in fact neglect the term $\frac{1}{2^{H+1}-1}$ on the right-hand-side (RHS), since we assume that $2^{H+1} \gg 1$.

We divide the terms in the summation on the LHS of \eqref{key_ineq} into two cases, as follows. Set $l^* := \left\lfloor\frac{\log(\frac{1}{1+\rho})}{\log p}\right\rfloor$. For simplicity, assume that $p$ and $\rho$ are such that $\frac{\log(\frac{1}{1+\rho})}{\log p}$ is not an integer. (We will later deal with the situation when this assumption does not hold.) Then, we either have $l \le l^*$ (Case~1), in which case $p^l > \frac{1}{1+\rho}$, or $l > l^*$ (Case~2), in which case $p^l < \frac{1}{1+\rho}$.
\begin{itemize}
	\item \textbf{Case 1:} $0 \le l \le l^*$, i.e., $p^l>\frac{1}{1+\rho}$. Using the Chernoff bound for the sum of $n$ independent Bernoulli random variables (see e.g.\ \cite{boucheron2013concentration}), we get
	\begin{IEEEeqnarray}{rCl}\label{term1}
		\mathbb{P}(Z_l\geq k) &\ge&   1-\mathbb{P}(Z_l \le k)\nonumber\\
&\geq& 1-\exp\left[-k(1+\rho)\mathbb{D}\left(\frac{1}{1+\rho} \parallel p^{l}\right)\right], \ \ \ \ 
	\end{IEEEeqnarray}
 where $\mathbb{D}(\cdot \parallel \cdot)$ is the relative entropy, defined as $\mathbb{D}(r\parallel s)=r \log\frac{r}{s}+(1-r)\log\frac{1-r}{1-s}$. Thus, for all sufficiently large $k$, we have $P(Z_l \ge k) \approx 1$ for $l = 0,1,\ldots,l^*$.
 	\item \textbf{Case 2}: $l^* < l \le H-1$, i.e., $p^l<\frac{1}{1+\rho}$.	By virtue of Lemma \ref{lemma} and the Chernoff bound for the sum of $n$ independent Bernoulli random variables, we have
	\begin{IEEEeqnarray}{rCl}\label{term2}
		\mathbb{P}(Z_l\geq k)&\leq& \mathbb{P}(Z_{l^*+1}\geq k) \notag \\
		&=&  \exp\left[{-k(1+\rho) \mathbb{D}\left(\frac{1}{1+\rho} \parallel p^{l^*+1}\right)}\right]. \ \ \  \ 
	\end{IEEEeqnarray}
Thus, for all sufficiently large $k$, we have $P(Z_l \ge k) \approx 0$ for $l = l^*+1,\ldots,H-1$.
\end{itemize}
Using the above approximations, we obtain the following estimate for the sum in the LHS of \eqref{key_ineq}: for all sufficiently large $k$, 
\begin{align*}
\sum_{l=0}^{H-1}2^{l+1} & \mathbb{P}(Z_l\geq k) \\ & = \ \sum_{l=0}^{l^*}2^{l+1}\mathbb{P}(Z_l\geq k) +  \sum_{l=l^*+1}^{H-1}2^{l+1}\mathbb{P}(Z_l\geq k) \\
& \approx \ \sum_{l=0}^{l^*}2^{l+1} \ \ = \ \ 2^{l^*+2}-2.
\end{align*}
With this, the LHS of \eqref{key_ineq} becomes
$$
\frac{\sum_{l=0}^{H-1}2^{l+1}\mathbb{P}(Z_l\geq k)}{2^{H+1}-1} \ \approx \ \frac{2^{l^*+2}-2}{2^{H+1}-1} \ \approx \ 2^{l^*-H+1},
$$
using our assumption that $2^{H+1} \gg 1$. Thus, for all $k$ sufficiently large, the inequality \eqref{key_ineq} becomes
\begin{equation}
2^{l^*-H+1} \gtrsim 1-\delta 
\label{approx_ineq}
\end{equation}
The exponent in the LHS above is an integer, and for $0\leq\delta<\frac{1}{2}$,
the inequality is effectively the same as $2^{l^*-H+1} \ge 1$, which holds iff $l^*\ge H-1$. Replacing $l^*$ by $\left\lfloor\frac{\log(\frac{1}{1+\rho})}{\log p}\right\rfloor$, we find that $l^* \ge H-1$ iff $p \ge \left(\frac{1}{1+\rho}\right)^{\frac{1}{H-1}}$. 

At this point, we would like to comment on what happens to the analysis above when $\frac{\log(\frac{1}{1+\rho})}{\log p}$ is in fact an integer, for instance, when $p = \left(\frac{1}{1+\rho}\right)^{\frac{1}{H-1}}$.
In this case, $l^* = \frac{\log(\frac{1}{1+\rho})}{\log p}$, and the bound obtained in Case~1 would hold for all $l\leq l^*-1$, while the bound in Case~2 would hold for all $l\geq l^*+1$. For $l = l^* = \frac{\log(\frac{1}{1+\rho})}{\log p}$, the mean of $Z_{l^*}$ is $k(1+\rho)p^{l^*} = k$, and hence, for large $k$, we have $\P(Z_{l^*} \ge k) \approx \frac12$. This additional term needs to be accounted for on the LHS of \eqref{key_ineq}, but this does not result in a significant change in the ensuing analysis.

We can therefore conclude from the discussion following \eqref{approx_ineq} that the least value of $p$ for which the inequality in \eqref{key_ineq} holds is, effectively, 
\begin{equation}
p_{k,\rho,\delta} = \left(\frac{1}{1+\rho}\right)^{\frac{1}{H-1}}. 
\label{pkrho}
\end{equation}

	\begin{figure} 
		\centering
		\subfloat[Minimum retransmission probability]{%
			\includegraphics[width=\linewidth]{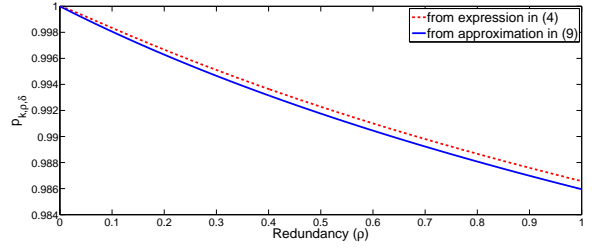}}
		\label{4a}
		\subfloat[Expected total number of transmissions]{%
			\includegraphics[width=\linewidth]{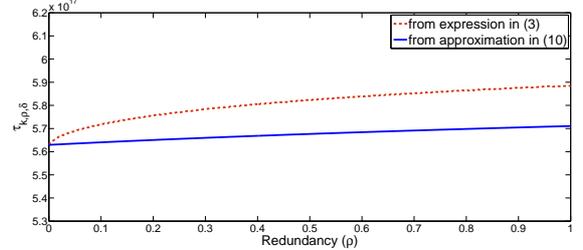}}
		\label{4b}\\
		\caption{Comparison between the true values of $p_{k,\rho,\delta}$ and $\tau_{k,\rho,\delta}$ obtained from (\ref{exprtomin}) and \eqref{taukrho}, and their corresponding expressions in \eqref{pkrho} and \eqref{taukrho2}, for $k = 500$, $\delta = 0.1$ and $H=50$.}
		\label{fig4} 
		\vspace{-1.5em}
	\end{figure}

Correspondingly, the expected total number of transmissions across all $N$ nodes of the tree is given, via \eqref{taukrho}, by
\begin{align}
\tau_{k,\rho,\delta}  \ & = \ k(1+\rho) \, \frac{2^H \left(\frac{1}{1+\rho}\right)^{\frac{H}{H-1}} - 1}{2\left(\frac{1}{1+\rho}\right)^{\frac{1}{H-1}}-1} \notag \\
& = \ k \, \frac{2^H \left(\frac{1}{1+\rho}\right)^{\frac{1}{H-1}} - (1 + \rho)}{2\left(\frac{1}{1+\rho}\right)^{\frac{1}{H-1}}-1} 
\label{taukrho2}
\end{align}

The expressions in \eqref{pkrho} and \eqref{taukrho2} have been plotted in Fig.~\ref{fig4} for $k = 500$ message packets and a binary tree of height $H = 50$. Also, plotted are the true values of $p_{k,\rho,\delta}$ obtained from the inequality \eqref{exprtomin} with $\delta = 0.1$, and the corresponding $\tau_{k,\rho,\delta}$ from \eqref{taukrho}. 	The plots allow us to conclude that there is no benefit in introducing redundancy in the form of coding to the probabilistic retransmission protocol on a tree, since the expected total number of transmissions increases with redundancy.

\section{Discussion}\label{sec:discuss}

\begin{figure}[!t]
	\centering
	\subfloat[31 $\times$ 31 grid ($G$)]{%
		\includegraphics[width=0.4\linewidth]{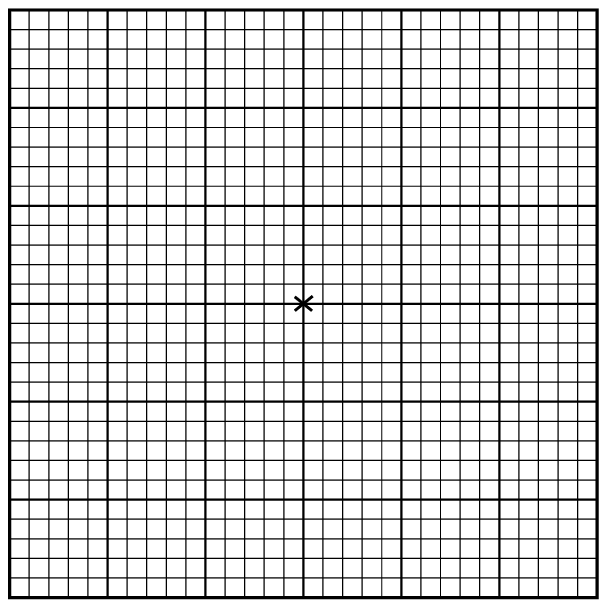}}
	\label{5a}
	\subfloat[Every fifth row ($G5$)]{%
		\includegraphics[width=0.4\linewidth]{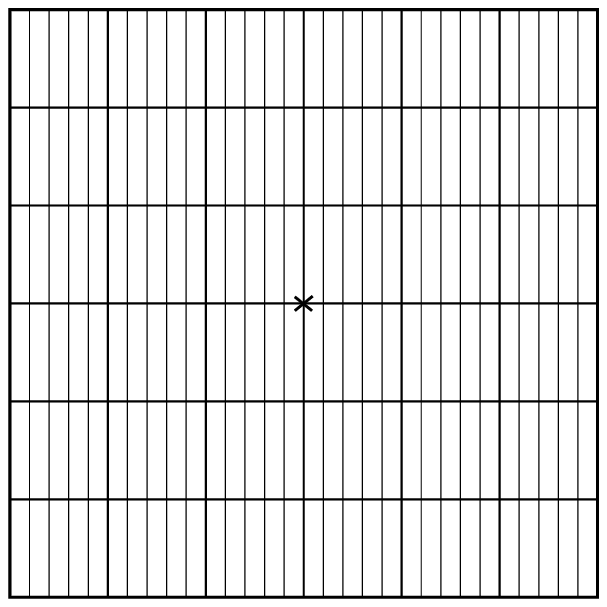}}
	\label{5b}
	\subfloat[Every tenth row ($G10$)]{%
		\includegraphics[width=0.4\linewidth]{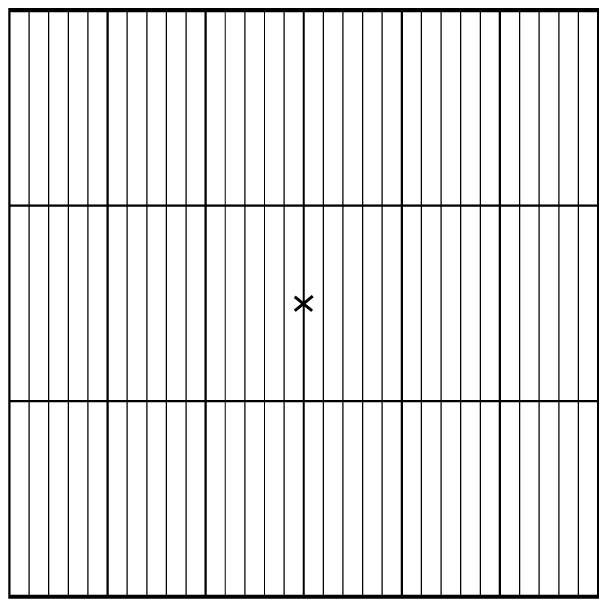}}
	\label{5c}
	\subfloat[Boundary and the center row ($G15$)]{%
		\includegraphics[width=0.4\linewidth]{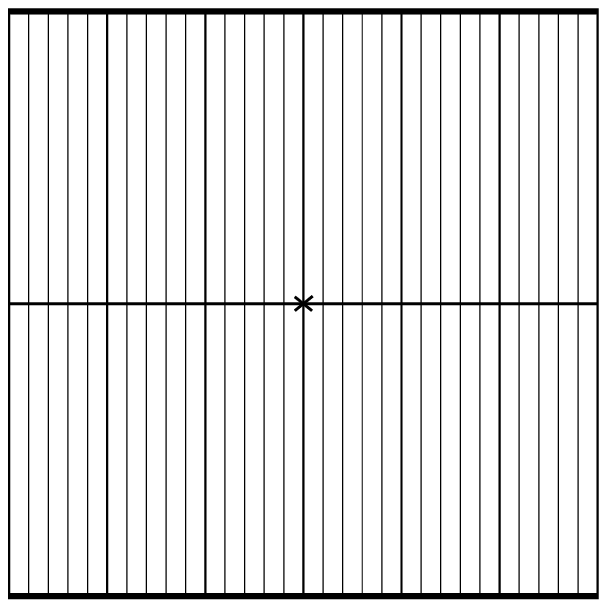}}
	\label{5d}\\
	\caption{Graphs to illustrate the importance of multiple paths.}
	\label{grids} 

\end{figure}

The simulations and analysis of the last two sections reveal that there is a significant benefit to introducing coding-based redundancy into the probabilistic retransmission protocol when the underlying network topology is highly connected (as in a large grid), but not so when the underlying network is a tree. The benefit is in terms of a reduction in the overall number of transmissions needed for a successful broadcast. This phenomenon seems to arise from the availability of ``multipath diversity'' in the network, i.e., the existence of multiple paths between the source node and any other node in the network. Indeed, in a binary tree, there is only one path from the root to any other node, whereas in a large grid, there is abundant multipath diversity. 

To test our multipath diversity hypothesis more systematically, we performed further simulations of the probabilistic forwarding protocol on graphs with different levels of multipath diversity. Starting with the $31 \times 31$ grid $G$ depicted in Fig.~\ref{grids}(a), we systematically deleted edges to obtain subgraphs $G5$, $G10$ and $G15$ with lower multipath diversity. Specifically, the graph $Gq$ (for $q = 5,10,15$) was obtained from the grid $G$ as follows. The nodes of $G$ form a $31 \times 31$ array, whose rows can be indexed by the integers $0,1,2,\ldots,30$, with $0$ denoting the index of the topmost row. Then, $Gq$ is obtained from $G$ by retaining the horizontal edges connecting adjacent nodes in row $j$, for every $j$ that is a multiple of $q$, and deleting all other horizontal edges --- see Figs.~\ref{grids}(b)--(d). The multipath diversity evidently decreases as $q$ increases. The results of our simulations, for $k=100$ packets, with the expected fraction of nodes receiving at least $k$ packets being $1-\delta = 0.9$, are shown in Fig. \ref{grids_simu}. In these simulations, the source node is the node at the centre of the grid, depicted by a `$\times$' in each of the graphs in Fig.~\ref{grids}. 

\begin{figure}[!t]
	\centering
	\subfloat[Minimum forwarding probability]{%
		\includegraphics[width=\linewidth]{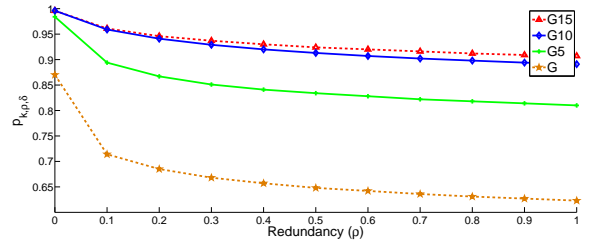}}
	\label{grids_simua}
	\medskip
	
	\subfloat[Expected total number of transmissions]{%
		\includegraphics[width=\linewidth]{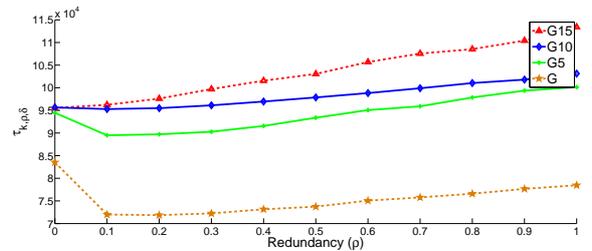}}
	\label{grids_simub}
	\vskip 8pt

	\caption{Simulation results on $G, G5, G10$ and $G15$.}
	\label{grids_simu} 

\end{figure}

\newpage
Clearly, the results plotted in Fig.~\ref{grids_simu} support our hypothesis that multipath diversity in a network is an important factor in determining whether or not the probabilistic forwarding protocol over the network would benefit from coding-based redundancy. In future work, we hope to be able to precisely characterize how multiple paths lead to a reduction in the overall number of transmissions needed for a successful broadcast.

\section*{Acknowledgements} The research presented in this paper was supported in part by the DRDO-IISc ``Frontiers'' Research Programme, and by a Cisco PhD Fellowship awarded to the first author.

\bibliographystyle{IEEEtran}
\bibliography{NCCpaper}

\end{document}